\documentclass[a4paper,11pt]{article}

\usepackage[margin=30mm]{geometry}
\usepackage{cite} 
\usepackage{enumerate} 
\usepackage{amsmath,amsfonts,amssymb,amsthm}
\usepackage{graphicx,epsfig,wrapfig}
\usepackage[colorlinks=true,citecolor=black,linkcolor=black,urlcolor=blue]{hyperref}
\usepackage{color,xcolor}
\usepackage{microtype} 
\usepackage{fancybox}
\usepackage{stmaryrd}
\usepackage{clrscode3e}
\usepackage{varwidth}

\graphicspath{{figs/}}

\newtheorem{theorem}{Theorem}

\newtheorem{lemma}[theorem]{Lemma}

\newcommand{\RR}{\ensuremath{\mathbb R}}  
\newcommand{\NN}{\ensuremath{\mathbb N}}  
\newcommand{\LL}{\ensuremath{\mathbb L}}  %
  
\newcommand{\D}{\ensuremath{\mathcal{D}}}  
  
\newcommand{\GG}{\ensuremath{G_{\le 1}}}

\def\dist{\mathit{dist}}
\def\best{\mathit{best}}
\newcommand{\cycle}{\mathit{cycle}}
\newcommand{\walk}{\mathit{walk}}
\newcommand\CR{\mbox{\tt cr}_2}		  

\def\DEF#1{\textbf{\emph{#1}}}
 
\let\leq\leqslant

\let\le\leqslant
\let\ge\geqslant

\newcommand{\out}[1]{}

\title{Two Optimization Problems for Unit Disks\thanks{Supported by the Slovenian Research Agency, core program P1-0297 and project L7-5459.}}

\author{Sergio Cabello\footnote{FMF, University of Ljubljana, and 
	Institute of Mathematics, Physics and Mechanics, Slovenia.}
\and
	Lazar Milinkovi\'c\footnote{FMF and FRI, University of Ljubljana, Slovenia.}}

\begin{document}
\maketitle

\begin{abstract}
We present an implementation of a recent algorithm to compute shortest-path trees
in unit disk graphs in $O(n\log n)$ worst-case time, 
where $n$ is the number of disks.

In the minimum-separation problem, we are given $n$ unit disks and two points $s$ and $t$,
not contained in any of the disks, and we want to compute 
the minimum number of disks one needs to retain so that any curve connecting
$s$ to $t$ intersects some of the retained disks.
We present a new algorithm solving this problem in $O(n^2\log^3 n)$ worst-case time
and its implementation.
\end{abstract}

\section{Introduction}

In this paper we consider two geometric optimization problems in the plane
where unit disks play a prominent role. For both problems 
we discuss efficient algorithms to solve them, provide an implementation
of these algorithms, and present experimental results on the implementation.

The first problem we consider is computing a \emph{shortest-path tree} 
in the (unweighted) intersection graph of unit disks. 
The input to the problem is a set $\D$ of $n$ disks of the same size, 
each disk represented by its center.
The corresponding unit disk (intersection) graph has a vertex for each disk,
and an edge connecting two disks $D$ and $D'$ of $\D$ whenever $D$ and $D'$ intersect.
An alternative, more convenient point of view, is to take as vertex set the set of
centers of the disks, denoted by $P$, and connecting two points $p$ and $q$ of $P$ 
whenever the Euclidean length $|pq|$ is at most the diameter of a disk. 
The graph is unweighted.
Given a root $r\in P$, the task is to compute a shortest-path tree from $r$ in this graph.
See Figure~\ref{fig:example1}.

\begin{figure}[htb]
	\centering
	\includegraphics[page=4,width=.32\textwidth]{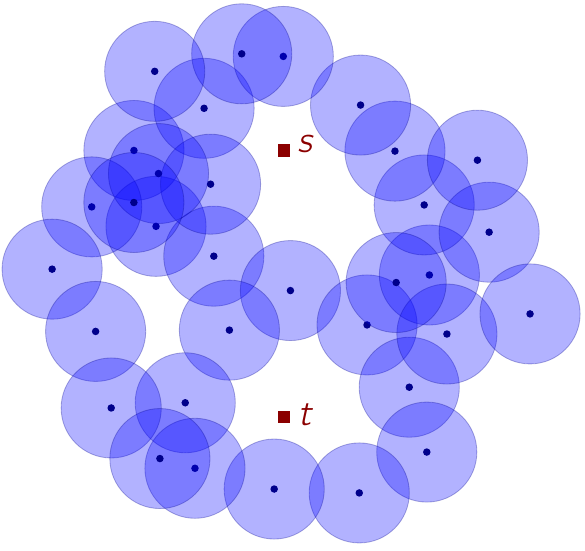}
	\hfill
	\includegraphics[page=2,width=.32\textwidth]{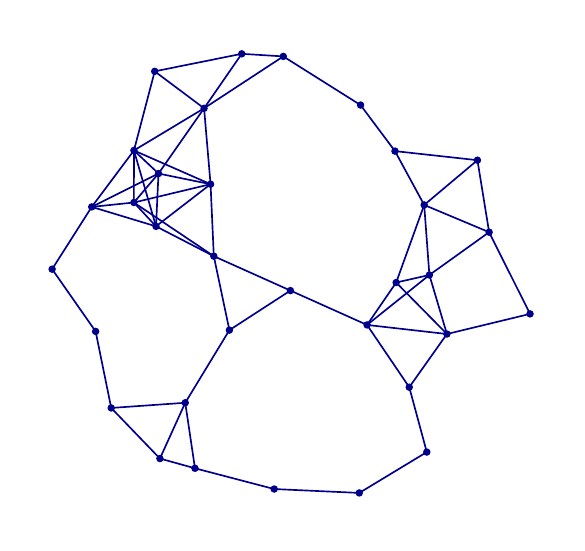}
	\hfill
	\includegraphics[page=3,width=.32\textwidth]{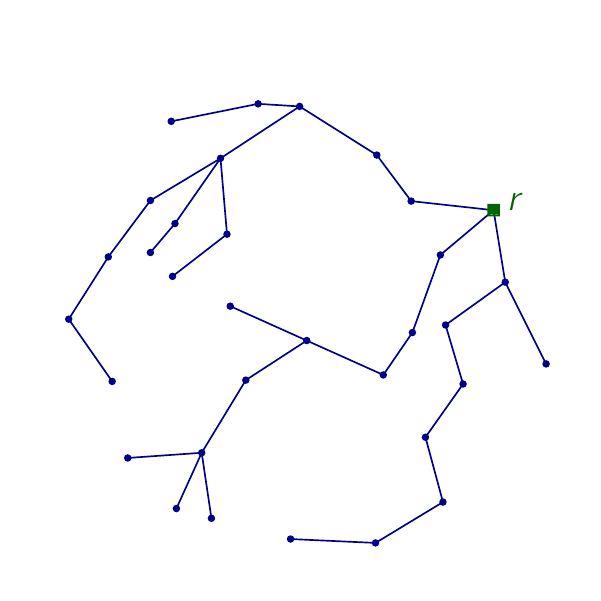}
	\caption{Left: unit disks and two additional points $s$ and $t$. 
		Middle: intersection graph of the disks. 
		Right: a shortest-path tree in the graph.}
	\label{fig:example1}
\end{figure}

The second problem we consider is the \emph{minimum-separation problem}.
The input is a set $\D$ of $n$ unit disks in the plane and two points
$s$ and $t$ not covered by any disks of $\D$. 
We say that $\D$ \emph{separates} $s$ and $t$ if each curve in the plane 
from $s$ to $t$ intersects some disk of $\D$.
The task is to find the minimum cardinality subset of $\D$
that separates $s$ and $t$. See the left of Figure~\ref{fig:example1} 
for an example of an instance.
Formally, we want to solve
\begin{align*}
	\min ~~		& |\D'|\\
	 \mbox{s.t.}~~ & \D'\subseteq \D \text{ and $\D'$ separates $s$ and $t$}. 
\end{align*}

Unit disks are the most standard model used for wireless sensor networks; 
see for example~\cite{GG11,HS95,zg-wsn-04}. 
Often the model is referred as UDG.
This model provides an appropriate trade off between simplicity and accuracy. 
Other models are more accurate, as for example discussed in~\cite{KWZ03,LP10},
but obtaining efficient algorithms for them is much more difficult.

While unit disks give a simple model, exploiting the geometric features 
of the model is often challenging. 
Shortest paths in unit disk graphs are essential for routing and
are a basic subroutine for several other more complex tasks. 
A somehow unexpected application of shortest paths in unit-disk graphs
to boundary recognition is given in~\cite{WGM06}.
The minimum-separation problem and variants thereof have been considered 
in~\cite{CG16,gkpvv-16}. 
The problem is dual to the barrier-resilience problem considered in~\cite{BK09,KH07,KLA07}.
It is not obvious that the minimum-separation problem can be solved optimally
in polynomial time, and the known algorithm for this uses as a subroutine 
shortest paths in unit disk graphs. 
Thus, both problems considered in this paper are related and it is worth to 
consider them together.

\paragraph{Our contribution}
We are aware of three algorithms to compute shortest-path trees in unit disk graphs in
$O(n\log n)$ worst-case time: one by Cabello and Jej\v{c}i\v{c}~\cite{CJ15}, 
one by Chan and Skrepetos~\cite{ChanS16}, and one Efrat, Itai and Katz~\cite{eik-01}. 
Here we report on an implementation of a modification of the algorithm in~\cite{CJ15},
and compare it against two obvious alternatives.
The only complex ingredients in the algorithm is computing the Delaunay triangulation
and static nearest-neighbour queries, but efficient libraries are available for this.
The algorithm of~\cite{eik-01} would be substantially harder to implement 
and it has worse constants hidden in the $O$-notation. The algorithm of~\cite{ChanS16}
for single source shortest paths is implementable and we expect that it would work
good in practice. However, this last algorithm has been published only
very recently, when we had completed our research.

As mentioned before, it is not obvious that the minimum-separation problem 
can be solved in polynomial time. 
In particular, the conference version~\cite{gkv-ipud-11} of~\cite{gkpvv-16} gave
2-approximation algorithm for the problem. 
Cabello and Giannopoulos~\cite{CG16} provide an exact algorithm that takes $O(n^3)$ 
worst-case time and works for arbitrary shapes, not just disks. 
In this paper we improve this last algorithm to near-quadratic time for the 
special case of unit disks. 
The basic principle of the algorithm is the same, but several additional tools
from Computational Geometry exploiting that we have unit disks
have to be employed to reduce the worst-case running time. 
Furthermore, we implement a variant of the new, near-quadratic-time algorithm and report
on the experiments.

\paragraph{Assumptions} 
We will assume that \emph{unit disk} means that it has radius $1/2$. 
Up to scaling the input data, this choice is irrelevant.
However, it is convenient for the exposition
because then the disks intersect whenever 
the distance between their centers is $1$. 
The implementation and the experiments also make this assumption.

Henceforth $P$ will be the set of centers of $\D$. 
All the computation will be concentrated on $P$. 
In particular, we assume that $P$ is known.
(For the shortest path problem, one could possibly consider weaker models based 
on adjacencies.)

We will work with the graph $\GG(P)$ with vertex set $P$ 
and an edge between two points $p,q\in P$ 
whenever their Euclidean distance $|pq|$ is at most $1$. 
In the notation we remove the dependency on $P$ and on the distance.
Thus we just use $G$ instead of $\GG(P)$.
For simplicity of the theoretical 
exposition we will sometimes assume that $G$ is connected.
It is trivial to adapt to the general case, for example
treating each connected component separately.
The implementation does not make this assumption.

\paragraph{Organization of the paper} 
In Section~\ref{sec:algorithms} we discuss the theoretical
algorithms for both problems and their guarantees.
In Section~\ref{sec:implementation} we discuss the implementations 
and the experimental results.

\section{Description of algorithms}
\label{sec:algorithms}

\subsection{Shortest-path tree in unit-disk graphs}
\label{sec:algorithm-sptree}
We describe here the algorithm of Cabello and Jej\v{c}i\v{c}~\cite{CJ15} 
to compute a shortest path tree in $G$ from a given root point $r\in P$. 
As it is usually done for shortest path algorithms, 
we use tables $\dist[\cdot]$ and $\pi[\cdot]$ indexed by the points of $P$ to record, 
for each point $p\in P$, the distance $d_{G}(s,p)$ and the ancestor of $p$ 
in a shortest $(s,p)$-path. 

The pseudocode of the algorithm, which we call \proc{UnweightedShortestPath},
 is in Figure~\ref{fig:BFS}. 
We explain the intuition, taken almost verbatim from~\cite{CJ15}.
We start by computing the Delaunay triangulation $DT(P)$ of $P$. 
We then proceed in rounds for increasing values of $i$, 
where at round $i$ we find the set $W_i$ of points at distance exactly $i$ in $G$ from the source $r$. We start with $W_0=\{ r\}$. 
At round $i$, we use $DT(P)$ to grow a neighbourhood around the points of $W_{i-1}$ 
that contains $W_{i}$. 
More precisely, we consider the points adjacent to $W_{i-1}$ in $DT(P)$ as candidate points for $W_{i}$. For each candidate point that is found to lie in $W_{i}$, we also take its adjacent vertices in $DT(P)$ as new candidates to be included in $W_{i}$. For checking whether a candidate point $p$ lies in $W_{i}$ we use a data structure to find a nearest neighbour of $p$ in $W_{i-1}$. If the distance from $p$ to its nearest neighbour $w$ in $W_{i-1}$ is
smaller than $1$, then the shortest path tree is extended by connecting $p$ to $w$.

\begin{figure}[htb]
\begin{center}
\ovalbox{~~~~
\begin{varwidth}{\linewidth}
\begin{codebox}
    \Procname{$\proc{UnweightedShortestPath}(P,r)$}
    \li build the Delaunay triangulation $DT(P)$
    \li \For$p\in P$ \Do
    	\li $\dist[p] \gets \infty$
    	\li $\pi[p] \gets \const{nil}$\End
    \li $\dist[r] \gets  0$
    \li $W_0 \gets \{ r\}$
    \li $i\gets1$
    \li \While $W_{i-1}\neq\emptyset$ \Do
    	\li build data structure for nearest neighbour queries in $W_{i-1}$
    	\li $Q \gets W_{i-1}$ \>\>\>\Comment candidate points
        \li $W_{i}\gets\emptyset$
    	\li \While $Q\neq\emptyset$\Do
    		\li $q$ an arbitrary point of $Q$
            \li remove $q$ from $Q$
            \li \For $qp$ edge in $DT(P)$ \Do
            	\li \If $\dist[p]=\infty$ \Then
					\li $w \gets$ nearest neighbour of $p$ in $W_{i-1}$
					\li \If $|pw|\leq 1$ \Then
						\li $\dist[p]\gets i$
						\li $\pi[p]\gets w$
						\li add $p$ to $Q$
						\li add $p$ to $W_{i}$
						\End
                    \End
            	\End
            \End
        \li $i\gets i+1$
    \End
    \li \Return $\dist[\cdot]$ and $\pi[\cdot]$
    \\[-2mm]
\end{codebox}
\end{varwidth}
~~~~}\end{center}
\caption{Algorithm from~\cite{CJ15} to compute a shortest path tree in the unweighted case.}
\label{fig:BFS}
\end{figure}

Cabello and Jej\v{c}i\v{c}~\cite{CJ15} show that the algorithm correctly computes
the shortest-path tree from $r$. 
If for nearest neighbors we use a data structure that,
for $n$ points, has construction time $T_c(n)$ and query time $T_q(n)$, 
and the Delaunay triangulation is computed in $T_{DT}(n)$ time,
then the algorithm takes $O(T_{DT}(n)+ T_c(n)+ n T_q(n))$ time. 
Standards tools in Computational Geometry imply that 
$T_{DT}(n)=O(n\log n)$, $T_c(n)=O(n\log n)$ and $T_q(n)=O(\log n)$.
This leads to the following.

\begin{theorem}[Cabello and Jej\v{c}i\v{c}~\cite{CJ15}]
  Let $P$ be a set of $n$ points in the plane and let $r$ be a point from $P$. 
  In time $O(n \log n)$ we can compute a shortest path tree from $r$
  in the unweighted graph $\GG(P)$.
\end{theorem}

It is clear that, when computing the shortest path tree from several sources,
we only need to compute the Delaunay triangulation once.

\subsection{Minimum separation with unit-disk}
\label{sec:algorithm-separation}

Cabello and Giannopoulos~\cite{CG16} present an algorithm
for the minimum separation problem that in the worst-case runs in cubic-time.
The algorithm has one feature that is both an advantage and a disadvantage: 
it works for any reasonable shapes, like segments or ellipses, and not just unit disks.
This means that it is very generic, which is good,
but it cannot exploit any properties of  unit disks.

In this section we are going to describe an algorithm to solve the minimum separation
problem \emph{for unit disks} in roughly quadratic time.
The improvement is based on 3 ingredients. 
The first ingredient is a reinterpretation of the algorithm of~\cite{CG16} 
for disks. In the original algorithm, we had to select a point inside each
shape. For disks there is a natural, obvious choice, the center of the disk.
This allows for a simpler description and interpretation of the algorithm.
We provide the description in Section~\ref{sec:generic}

The second ingredient is the efficient algorithm for shortest-path trees for the graph $G$.
The third ingredient is a compact treatment of the edges of $G$ using
a few tools from Computational Geometry, namely 
range trees, point-line duality, and nearest-neighbour searches.
This is explained in Section~\ref{sec:quadratic}.

\subsubsection{Generic algorithm specialized for unit disks}
\label{sec:generic}
Let us first introduce some notation.
Recall that $s$ and $t$ are the two points to separate.
Each walk $W$ in the graph $G=\GG(P)$ defines a planar polygonal curve
in the obvious way: we connect the points of $P$ 
with segments in the order given by $W$. 
We will relax the notation slightly and denote also by $W$ the curve itself.
For any spanning tree $T$ of $G$ and any edge $e\in E(G)\setminus E(T)$, 
let $\cycle(T,e)$ be the unique cycle in $T+e$.
Finally, for any walk in $G(P)$, let $\CR (st,W)$ be the 
modulo $2$ value of the number of crossings between the segment $st$ 
and (the curve defined by) $W$.
The following property is implicit in~\cite{CG16} and explicit in~\cite{CK15}:
\begin{quote}
	Let $T$ be any spanning tree of $G$.
	The set of unit disks with centers in $P$ separate $s$ and $t$ if and only
	if there exists some edge $e\in E(G)\setminus E(T)$ 
	such that $\CR (st, \cycle(T,e))=1$.
\end{quote}

A consequence of this is that finding a minimum separation amounts 
to finding a shortest cycle in $G$ that crosses the segment $st$ an odd number of times.
Moreover, one can show that we can restrict our search to 
a very concrete family cycles, as follows. 
Consider any optimal cycle $W^*$ and let $r^*$ be any vertex in $W^*$.
Fix a shortest-path tree $T_{r^*}$ from $r^*$ in $G$.
When there are many, the choice of $T_{r^*}$ is irrelevant. 
Then, the set of cycles 
\[
	\{ \cycle(T_{r^*},e)\mid e\in E(G)\setminus E(T_{r^*})\}
\]
contains an optimal solution.
This follows from the co-called 3-path condition.
We include here the key property that implies this claim and
spell out a self-contained proof. See~\cite{CG16} for very similar ideas.

\begin{lemma}
	Let $W^*$ be a shortest cycle in $G$ that crosses the segment 
	$st$ an odd number of times and let $r^*$ be any vertex in $W^*$.
	Fix a shortest-path tree $T_{r^*}$ from $r^*$ in $G$.
	Then, the set of cycles $\{ \cycle(T_{r^*},e)\mid e\in E(G)\setminus E(T_{r^*})\}$
	contains a shortest cycle of $G$ that crosses $st$ an odd number of times.
\end{lemma}
\begin{proof}
	For any points $p$ and $q$ of $P$, let $T_{r^*}[p \rightarrow q]$ be the 
	unique path contained in $T_{r^*}$ from $p$ to $q$.
	For every edge $pq$ of $G$, let $\walk(T_{r^*},pq)$ be the closed walk that follows 
	$T_{r^*}[r^*\rightarrow p]$, then the edge $pq$, and finally $T_{r^*}[q\rightarrow r^*]$.
	We then have the following relation modulo 2:
	\begin{align*}
		\sum_{pq\in W^*}& \CR(st,\walk(T_{r^*},pq)) \\ 
		&=~ 
		\sum_{pq\in W^*} \bigl( \CR(st,T_{r^*}[r^* \rightarrow p])+ \CR(st,pq) + 
						\CR(st,T_{r^*}[q \rightarrow r^*]) \bigr) \\
		&=~ \sum_{pq\in W^*} \CR(st,pq)  \\
		&=~ \CR(st,W^*) \\
		&=~ 1. 		
	\end{align*}
	In the second equality we have used that each path
	$T_{r^*}[r^*\rightarrow p]$ and its reverse $T_{r^*}[p\rightarrow r^*]$ 
	appears an even number of times in the sum,
	and thus cancel out modulo 2.
	Parity implies that, for some edge $p_0q_0$ of $W^*$, 
	we have $\CR(st,\walk(T_{r^*},p_0q_0))=1$.
	It must be that $p_0q_0\notin E(T_{r^*})$ because for each edge $pq$ of $T_{r^*}$
	it holds $\CR(st,\walk(T_{r^*},pq))=0$.
	
	Since $\CR(st,\walk(T_{r^*},p_0q_0))= \CR(st,\cycle(T_{r^*},p_0q_0))$ 
	because the path from $r^*$ to the lowest common ancestor of $p$ and $q$ in $T_{r^*}$ 
	is counted twice on the left side of the equality, 
	we have $\CR(st,\cycle(T_{r^*},p_0q_0))=1$. 

	Since $r^*$ is a vertex of $W^*$ and $p_0q_0$ is an edge of $W^*$,
	the length of $W^*$ is at least the length of $T_{r^*}[r^* \rightarrow p_0]$ plus $1$,
	for the edge $p_0q_0$,
	plus the length of $T_{r^*}[q_0 \rightarrow r^*]$. However, this second part
	is exactly the length of $\walk(T_{r^*},p_0q_0)$, which is at least the length of 
	$\cycle(T_{r^*},p_0q_0)$.

	We have shown that, for some edge $p_0q_0\in E(G)\setminus E(T_{r^*})$,
	the cycle $\cycle(T_{r^*},p_0q_0)$ is not longer than $W^*$ 
	and crosses $st$ an odd number of times. The result follows.
\end{proof}

Since we do not know a vertex $r^*$ in the shortest cycle of $G$, we just try all
possible roots as candidates. (This leads to the option of having a randomized algorithm,
by selecting some roots at random, for the case where the optimal solution is large.)
Thus, for each vertex $r$ of $G$, we fix a shortest-path tree $T_r$ from $r$ in $G$,
and then the size of the optimal solution is given by
\[
	\min \{ 1+d_G(r,p)+d_G(r,q) \mid
		r\in P,~ pq\in E(G)\setminus E(T_r),~
		\CR(st,\cycle(T_{r},pq))=1 \} .
\]

The values $\CR (st, \cycle(T_r,e))$ can be computed in constant amortized time per edge
with some easy bookkeeping, as follows. Consider a fixed tree $T_r$.
For each point $p\in P$ we store $N[p]$ as the parity of the number of crossings
of the path in $T_r$ from $r$ to $p$. When $p$ is not the root,
the value $N[p]$ can be computed from the value of its parent $\pi[p]$ in $T_r$
using that $N[p]=N[\pi[p]]+\CR(st,p\pi[p])$.
In the algorithm we have written it this way (lines 4--6), but
one can also compute the values at the time of computing the shortest path tree $T_r$.

We then have for each shortest-path tree $T_r$
\begin{align*}
	\forall pq\in E(G)\setminus E(T_r):&~~~~
	\CR (st, \cycle(T_r,pq))= N[p]+N[q]+\CR(st,pq) \pmod 2\\
	\forall pq\in E(T_r):&~~~~
	0= N[p]+N[q]+\CR(st,pq) \pmod 2
\end{align*}
because crossings that are counted twice cancel out modulo $2$.
In particular, the path in $T_r$ from $r$ to the lowest common ancestor of $p$ and $q$
is counted twice.
This implies that we can just check for \emph{all} edges $pq$ of $G$ whether
the sum  $N[p]+N[q]+\CR(st,pq)$ is $0$ modulo $2$.
The final resulting algorithm, denoted as \proc{GenericMinimumSeparation},
is given in Figure~\ref{fig:generic}.

\begin{figure}[htb]
\begin{center}
\ovalbox{~~~~
\begin{varwidth}{\linewidth}
\begin{codebox}
    \Procname{$\proc{GenericMinimumSeparation}(P,s,t)$}
    \li $\best \gets \infty$ // length of the best separation so far
    \li \For $r\in P$ \Do
    	\li $(\dist[~],\pi[~])\gets$ shortest path tree from $r$ in $G(P)$
		\\ \> \Comment Compute $N[~]$
    	\li $N[r]=0$
    	\li \For $p\in P\setminus \{r\}$ in non-decreasing values of $\dist[p]$ \Do
			\li $N[p]= N[\pi[p]]+ \CR(st,p\pi[p]) \pmod 2$
			\End
    	\li \For $pq\in E(G(P))$\Do
			\li \If $N[p]+N[q]+ \CR(pq,st) \pmod 2 =1$ \Then
				\li $\best \gets \min \{ \best, \dist[p]+\dist[q]+1\}$ 
				\End
			\End
		\End
    \li \Return $\best$
    \\[-2mm]
\end{codebox}
\end{varwidth}
~~~~}\end{center}
\caption{Adaptation of the generic algorithm to compute the minimum separation for unit disks.}
\label{fig:generic}
\end{figure}

Let us look into the time complexity of the algorithm.
For each point $r\in P$ we have to compute a shortest-path tree in $G$.
This can be done in $O(n\log n)$ in our case, 
as discussed in Section~\ref{sec:algorithm-sptree}. 
Then, for each edge $pq$ of $G$ some constant amount of work is done.
Thus for each point $r$ we spend  
$O(n\log n+|E(G)|)$.
This is cubic in the worst-case.
We could get an improved running time if we can treat all the edges of $G$
compactly. This is what we explain next.

\subsubsection{Compact treatment of edges}
\label{sec:quadratic}

From now on we will assume that $s$ is the origin and 
$t$ is the point $(0,\tau)$, with $\tau\ge 0$. Thus, the segment $st$
is vertical and $t$ is above $s$. The implementation just
assumes that $st$ is vertical with $s$ below $t$. 
A simple rigid transformation can be applied to
the input to get to this setting.

We will use the data structure in the following lemma.
It is essentially a multi-level data structure consisting
of a 2-dimensional range tree $T$ with
a data structure for nearest neighbour at each node of the secondary structure
of $T$.

\begin{lemma}
\label{lem:block}
	Let $B$ be a set of $n$ points with positive $x$-coordinates.
	We can preprocess $B$ in $O(n\log^3 n)$ time such that,
	for any query point $a$	with negative $x$-coordinate, 
	we can decide in $O(\log^3 n)$ time whether the set 
	$\{ b\in B \mid \text{$ab$ intersects $\sigma$ and $|ab|\le 1$}\}$
	is empty.
	The same data structure can handle queries
	to know whether the set
	$\{ b\in B \mid \text{$ab$ does not intersect $\sigma$ and $|ab|\le 1$}\}$
	is empty.
\end{lemma}
\begin{proof}
	We are going to use point-line duality and range trees.
	These are standard concepts in Computational Geometry;
	see for example~\cite[Chapters 5 and 8]{bkos-08}.
	We assume that the reader is familiar with the topic.
	Figure~\ref{fig:duality} may be helpful in the following discussion.

	\begin{figure}[htb]
		\centering
		\includegraphics[width=\textwidth]{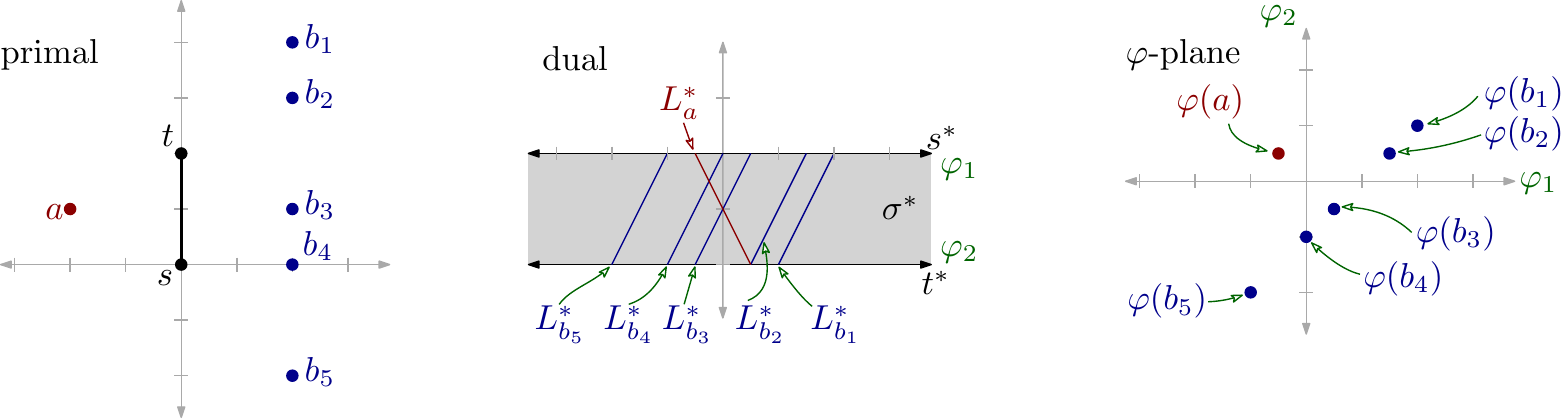}
		\caption{Transformation in the proof of Lemma~\ref{lem:block}.}
		\label{fig:duality}
	\end{figure}
	
	We use the following precise point-line duality: 
	the non-vertical line $\ell \equiv y=mx+c$ 
	is mapped to the point $\ell^*=(m,-c)$ and
	vice-versa. 
	Let $\LL$ be the set of non-vertical lines.
	Let $\sigma$ be the line segment $st$.
	Let $\sigma^*$ be the set of points dual to non-vertical lines that intersect $\sigma$.
	Thus
	\[
		\sigma^* ~=~ \{ l^* \mid \ell\in \LL, \ell\cap \sigma\neq \emptyset\}.
	\]
	Since we assumed that $s=(0,0)$ and $t=(0,\tau)$, 
	in the dual space $\sigma^*$ is the horizontal slab 
	\[
		\sigma^* ~=~ \{ (m,-c)\in \RR^2\mid 0\le c\le \tau\}.
	\]
	
	For every point $p\in \RR^2$, outside the $y$-axis, 
	let $L^* _p$ be the set of points dual
	to the lines through $p$ that intersect $\sigma$:
	\[
		L^*_p=\{ \ell^* \mid \ell\in \LL, p \in \ell, \text{ and } \sigma\cap \ell\not= \emptyset\}.
	\]
	In the dual space, $L^*_p$ is a segment with endpoints 
	$(\varphi_1(p),0)$ and $(\varphi_2(p),-\tau)$, 
	for some values $\varphi_1(p)$ and $\varphi_2(p)$
	that are easily computable.
	Namely, $\varphi_1(p)$ is the slope of the line through $p$
	and $(0,0)$ while $\varphi_2(p)$ is the slope of the line through $p$ and $(0,\tau)$.
	The segment $L^*_p$ is contained in the slab $\sigma^*$ and has the endpoints on
	different boundaries of $\sigma^*$.
	Finally, define the mapping $\varphi(p)=(\varphi_1(p),\varphi_2(p))$.
	Thus, $\varphi$ maps points in the plane with nonzero $x$-coordinate
	to points in the plane. 

	Let $a$ be any point to the left of the $y$-axis and 
	let $b$ be a point to the right of the $y$-axis.
	The segment $ab$ intersects $\sigma$ if and only if $L^*_a$ intersects $L^*_b$.
	Namely, an intersection of $L^*_a$ and $L^*_b$ is dual to the line
	through $a$ and $b$. The segments $L^*_a$ and $L^*_b$ intersect if
	and only if the order of their endpoints on the boundaries of $\sigma^*$ are reversed.
	Moreover, since $a$ is to the left of the $y$-axis and $b$ is to the right
	of the $y$-axis, if the segment $ab$ intersects $\sigma$,
    then $\varphi_1(a)$, the slope of the line through $a$ and $(0,0)$,
    is smaller than $\varphi_1(b)$, the slope of the line through $b$ and $(0,0)$.
	Thus we have the following property:
	\[
		ab \cap \sigma \neq \emptyset ~\Longleftrightarrow ~ 
		\varphi_1(a)\le \varphi_1(b)) \text{ and } \varphi_2(a)\ge \varphi_2(b)).
	\]			
	Given a point $a$ to the left of the $y$ axis, 
	the set of points $b\in B$ with the property
	that $ab$ intersects $\sigma$ corresponds to the points $b$ with $\varphi(b)$
	in the bottom-right quadrant with apex $\varphi(a)$. 	

	We can use a $2$-dimensional range tree to store the point set $\varphi(B)$,
	where each point $b\in B$ is identified with its image $\varphi(b)$. 
	Moreover, for each node $v$ in the secondary level of the range tree, we 
	store a data structure for nearest neighbours for the canonical set $P(v)$ of points
	that are stored below $v$ in the secondary structure.
	
	For any query $a\in A$, the points $b\in B$ such that $ab$ intersects
	$\sigma$ are obtained by querying the 2-dimensional range tree for the points
	of $\varphi(B)$ contained in the quadrant 
	\[
		\{(x,y)\mid  \varphi_1(a) \le x \text{ and } \varphi_2(a) \ge y\}.
	\]
	This means that we get the set
	$\{ b\in B \mid \text{$ab$ intersects $\sigma$}\}$
	as the union of canonical subsets 
	$P(v_1),\dots,P(v_k)$ for $k=O(\log^2 n)$ nodes in the secondary levels
	of the 2-dimensional range tree.
	For each such canonical subset $P(v_i)$, we query for the nearest neighbour of $a$.
	If for some $v_i$ we get a nearest neighbour at distance at most $1$ from $a$,
	then we know that $\{ b\in B \mid \text{$ab$ intersects $\sigma$ and $|ab|\le 1$}\}$
	is non-empty. Otherwise the set is empty.

	The construction time of the 2-dimensional range tree is $O(n\log n)$.
	Each point appears in $O(\log^2 n)$ canonical subsets $P(v)$.
	This means that $\sum_v |P(v)| = O( n\log^2 n)$, where the sum iterates over
	all nodes $v$ in the secondary data structure.
	Since for each node $v$ in the secondary level we build a data structure for
	nearest neighbours, which takes $O(|P(v)| \log |P(v)|)$,
	the total construction time is $O( n\log^3 n)$.
	For the query time, the standard 2-dimensionsal range tree takes
	$O(\log^2 n)$ time to find the $O(\log^2 n)$  nodes $v_1,\dots, v_k$ such
	that 
	\[
		\bigcup_{i=1}^k P(v_i) ~=~ \{ b\in B \mid \text{$ab$ intersects $\sigma$}\},
	\]
	and then we need additional $O(\log n)$ time per node to query for a nearest neighbor.

	Answering the queries for 
	$\{ b\in B \mid \text{$ab$ does not intersect $\sigma$ and $|ab|\le 1$}\}$
	is done similarly (and the same data structure works), we just have
	to query for 2 of the other quadrants. 
	(The top-left quadrant of $\varphi(a)$ is always empty.)
\end{proof}

Inside the data structure of Lemma~\ref{lem:block} we are using a data structure
for nearest neighbours with construction time $O(n\log n)$ and query time $O(\log n)$.
If we would use another data structure for nearest neighbours with construction
time $T_c(n)$ and query time $T_q(n)$, then the construction time 
in Lemma~\ref{lem:block} becomes $O(T_c(n\log^2 n))$ and the query time is
$O(T_q(n)\cdot \log^2 n)$.

From the theoretical perspective is would be more efficient to compute the union
\[
	\bigcup_{b\in B} \{ (x,y)\in \RR^2\mid x<0,~ |(x,y)b|\le 1,~ (x,y) 
			\text{ intersects } \sigma \}
\]
and make point location there. Since the regions cannot have many crossings,
good asymptotic bounds can be obtained. However, such approach seems to be 
only of theoretical interest and the improvement on the overall result
is rather marginal.

Consider now a fixed root $r$. Assume that we have computed 
the shortest path tree $T_r$ and the corresponding tables $\pi[~]$, $\dist[~]$ and $N[~]$,
as discussed in Section~\ref{sec:generic}.
We group the points by their distance from $r$:
\[
	W_i ~=~ \{ p\in P \mid \dist[p]=i \},~~~ i=0,1,\dots
\]
A standard property of BFS trees, that also holds here,
is that all the distances from the root for any two adjacent vertices differ by at most $1$.
That is, the neighbours of a point $p\in P$ in $G$ are contained in
$W_{\dist[p]-1}\cup W_{\dist[p]} \cup W_{\dist[p]+1}$.
We will exploit this property.

We make groups $L_i^j$ and $R_i^j$ (where $L$ stands for left and $R$ for right) 
defined by
\begin{align*}
	L_i^j ~&=~ \{ p\in P \mid \dist[p]=i,~ p.x<0,~ N[p]=j \} ,
	~~~\text{where $j=0,1$ and $i=0,1,\dots$}\\
	R_i^j ~&=~ \{ p\in P \mid \dist[p]=i,~ p.x>0,~ N[p]=j \} ,
	~~~\text{where $j=0,1$ and $i=0,1,\dots$}
\end{align*}
We are interested in edges $pq$ of $G$ such that $N[p]+N[q]+\CR(st,pq)=1 \pmod 2$. 
Up to symmetry (exchanging $p$ and $q$), this is equivalent to pairs of points
$(p,q)$ in one of the following two cases:
\begin{itemize}
	\item for some $i\in \NN$ and some $j\in \{0,1\}$, we have
			$p\in L_i^j\cup R_i^j$, 
			$q\in L_i^{1-j}\cup R_i^{1-j}\cup L_{i-1}^{1-j}\cup R_{i-1}^{1-j}$, 
			$|pq|\le 1$, and $pq$ does not cross $st$;
	\item for some $i\in \NN$ and some $j\in \{0,1\}$, we have
			$p\in L_i^j\cup R_i^j$, 
			$q\in L_i^{j}\cup R_i^{j}\cup L_{i-1}^{j}\cup R_{i-1}^{j}$, 
			$|pq|\le 1$, and $pq$ crosses $st$.
\end{itemize}
Each one of these cases can be solved efficiently.
Up to symmetry, we have the following cases:
\begin{itemize}
\item If we want to search the candidates $(p,q)\in L_i^j\times L_{i'}^{1-j}$
	(that cannot cross $st$ since they are on the same side
	of the $y$-axis), we first preprocess $L_{i'}^{1-j}$ for nearest neighbours.
	Then, for each point $p$ in $L_i^j$, we query the data structure
	to find its nearest neighbour $q_p$ in $L_i^j$.
	If for some $p$ we get that $|pq_p|\le 1$, then we have
	obtained an edge $pq_p$ of $G$ with $\CR(\cycle(T_r,pq_p))=1$
	and $\dist[p]+\dist[q_p]+1=i+i'+1$.
	If for each $p$ we have $|pq_p|> 1$, then $L_i^j\times L_{i'}^{1-j}$
	does not contain any edge of $G$.
	The overall running time, if $m=|L_i^j|+|L_{i'}^{1-j}|$, is
	$O(m\log m)$.
\item If we want to search the candidates $(p,q)\in L_i^j\times R_{i'}^{j}$
	such that $pq$ crosses $st$,
	we first preprocess $R_{i'}^{1-j}$ as discussed in Lemma~\ref{lem:block}
	into a data structure.
	Then, for each point $p\in L_i^j$ we query the data structure (for crossing $st$).
	If we get some nonempty set, then there is 
	an edge $pq$ of $G$ with $p\in L_i^j$, $q\in R_{i'}^{j}$, $\CR(\cycle(T_r,pq))=1$
	and $\dist[p]+\dist[q]+1=i+i'+1$.
	Otherwise, there is no edge $pq\in L_i^j\times R_{i'}^{j}$ that crosses $st$.
	The overall running time, if $m=|L_i^j|+|R_{i'}^{j}|$, is
	$O(m\log^3 m)$.
\item If we want to search the candidates $(p,q)\in L_i^j\times R_{i'}^{1-j}$
	such that $pq$ does not cross $st$,
	we first preprocess $R_{i'}^{1-j}$ as in Lemma~\ref{lem:block}
	into a data structure.
	Then, for each point $p\in L_i^j$ we query the data structure (for not crossing $st$).
	The remaining discussion is like in the previous item.
\end{itemize}
We conclude that each of the cases can be done in $O(m\log^3 m)$ worst-case time, where
$m$ is the number of points involved in the case.
Iterating over all possible values $i$, 
it is now easy to convert this into an algorithm that spends $O(n\log^3 n)$ time
per root $r$. We summarize the result we have obtained. This improves for the
case of unit disks the previous, generic algorithm.

\begin{theorem}
	The minimum-separation problem for $n$ unit disks 
	can be solved in $O(n^2 \log^3 n)$ time.
\end{theorem}
\begin{proof}
	Let $P$ be the centers of the disks and,
	as before, consider the graph $G=\GG(P)$.
	For each root $r\in P$ we build the shortest-path tree
	and the sets $W_i,L_i^0,L_i^1,R_i^0,L_i^1$ for all $i$ in $O(n\log n)$ time.
	We then have at most $n$ iterations where, at iteration $i$,
	we spend $O(|W_i\cup W_{i-1}|\log^3 |W_i\cup W_{i-1}|)$ time.
	Since the sets $W_i$ are disjoint,
	adding over $i$, this means that we spend $O(n\log^3 n)$ time per root $r\in P$.
	
	Correctness follows from the foregoing discussion and the fact 
	that the algorithm is computing the same as the generic algorithm.
\end{proof}

\out{
\begin{figure}[htb]
\begin{center}
\ovalbox{~~~~
\begin{varwidth}{\linewidth}
\begin{codebox}
    \Procname{~~\Comment Work for root $r\in P$}
   	\li $(\dist[~],\pi[~])\gets$ shortest path tree from $r$ in $\GG(P)$
	\li Compute the levels $W_0,W_1,\dots$
	\li \For $i=0\dots n$ \Do
		\li Compute $N[p]$ for each $p$ in $W_i$
		\li Compute $L^0_i$, $L^1_i$, $R^0_i$, $R^1_i$
		\End
    \li $i=1$
    \li \While $2i< \best$ and $W_i\neq\emptyset$ \Do
		\\ \> \Comment within each side of the $y$-axis 
		\li search candidates in \\
			\> ~~~$L^0_i\times L^1_{i-1}$, $L^1_i\times L^0_{i-1}$, 
				$R^0_i\times R^1_{i-1}$, $R^1_i\times R^0_{i-1}$,
				$L^0_i\times L^1_{i}$ and $R^0_i\times R^1_{i}$
		\\ \> \Comment across $y$-axis crosing $\sigma$
		\li search candidates crossing $\sigma$ in \\
			\> ~~~ $L^0_i\times R^0_{i-1}$, 
				$L^1_i\times R^1_{i-1}$, $L^0_{i-1}\times R^0_{i}$, 
				$L^1_{i-1}\times R^1_{i}$, $L^0_i\times R^0_{i}$ 
				and $L^1_i\times R^1_{i}$
		\\ \> \Comment across $y$-axis not crosing $\sigma$			
		\li search candidates not crossing $\sigma$ in \\
			\> ~~~ $L^0_i\times R^1_{i-1}$, 
				$L^1_i\times R^0_{i-1}$, $L^1_{i-1}\times R^0_{i}$, 
				$L^0_{i-1}\times R^1_{i}$, $L^0_i\times R^1_{i}$ 
				and $L^1_i\times R^0_{i}$		
		\li $i \gets i+1$
		\End
    \\[-2mm]
\end{codebox}
\end{varwidth}
~~~~}\end{center}
\caption{Work for each vertex in the new algorithm for minimum separation with unit disks.}
\label{fig:compact}
\end{figure}
}

\begin{figure}[htb]
\begin{center}
\scalebox{.85}{%
\ovalbox{~~~~
\begin{varwidth}{\linewidth}
\begin{codebox}
    \Procname{$\proc{SeparationUnitDisksCompact}(P,s,t)$}
    \li $\best \gets n+1$ // length of the best separation so far
    \li \For $r\in P$ \Do
    	\li $(\dist[~],\pi[~])\gets$ shortest path tree from $r$ in $\GG(P)$
		\\ \> \Comment Compute the levels $W_i$
    	\li \For $i=0\dots n$ \Do
			\li $W_i\gets$ new empty list
			\End
    	\li \For $p\in P$ \Do
			\li add $p$ to $W_{\dist[p]}$
			\End		
		\\ \> \Comment Compute $N[~]$ for the elements of $W_i$ and
		\\ \> \Comment and construct $L_i^0,L_i^1,R_i^0,R_i^1$
    	\li $N[r]=0$
		\li \For $i=1\dots n$ \Do
			\li \For $p\in W_i$ \Do
				\li $N[p]= N[\pi[p]]+ \CR(st,p\pi[p]) \pmod 2$
				\li \If $p$ to the left of the $y$-axis \Then
					\li add $p$ to $L^{N[p]}_i$
					\End
				\li \If $p$ to the right of the $y$-axis \Then
					\li add $p$ to $R^{N[p]}_i$
					\End
				\End
			\End
    	\li $i=1$
    	\li \While $2i< \best$ and $W_i\neq \emptyset$ \Do
			\\ \>\> \Comment length $2i$; within each side of the $y$-axis 
			\li search candidates in $L^0_i\times L^1_{i-1}$
			\li search candidates in $L^1_i\times L^0_{i-1}$
			\li search candidates in $R^0_i\times R^1_{i-1}$
			\li search candidates in $R^1_i\times R^0_{i-1}$
			\\ \>\> \Comment length $2i$; across $y$-axis crosing $\sigma$
			\li search candidates in $L^0_i\times R^0_{i-1}$ crossing $\sigma$
			\li search candidates in $L^1_i\times R^1_{i-1}$ crossing $\sigma$
			\li search candidates in $L^0_{i-1}\times R^0_{i}$ crossing $\sigma$
			\li search candidates in $L^1_{i-1}\times R^1_{i}$ crossing $\sigma$
			\\ \>\> \Comment length $2i$; across $y$-axis not crosing $\sigma$			
			\li search candidates in $L^0_i\times R^1_{i-1}$ not crossing $\sigma$
			\li search candidates in $L^1_i\times R^0_{i-1}$ not crossing $\sigma$
			\li search candidates in $L^0_{i-1}\times R^1_{i}$ not crossing $\sigma$
			\li search candidates in $L^1_{i-1}\times R^0_{i}$ not crossing $\sigma$
			\\ \>\> \Comment length $2i+1$; within each side of the $y$-axis 
			\li search candidates in $L^0_i\times L^1_{i}$
			\li search candidates in $R^0_i\times R^1_{i}$
			\\ \>\> \Comment length $2i+1$; across $y$-axis crosing $\sigma$
			\li search candidates in $L^0_i\times R^0_{i}$ crossing $\sigma$
			\li search candidates in $L^1_i\times R^1_{i}$ crossing $\sigma$
			\\ \>\> \Comment length $2i+1$; across $y$-axis not crosing $\sigma$			
			\li search candidates in $L^0_i\times R^1_{i}$ not crossing $\sigma$
			\li search candidates in $L^1_i\times R^0_{i}$ not crossing $\sigma$
			\li $i \gets i+1$
			\End
		\End
    \li \Return $\best$
    \\[-2mm]
\end{codebox}
\end{varwidth}~~~~}}
\end{center}
\caption{New algorithm for minimum separation with unit disks.}
\label{fig:FULLAlgorithm}
\end{figure}

The resulting new algorithm is given in Figure~\ref{fig:FULLAlgorithm}.
As before, the variable $\best$ stores the length of the shortest cycle (or actually
rooted closed walk) that we have found so far. We can start setting $\best=n+1$
at start.
If eventually we finish with the value $\best=n+1$, it means that there is no
feasible solution for the separation problem.
When we consider a root $r$ we are interested in closed walks rooted at $r$
and length at most $\best$. Since any closed walk through a vertex of $W_i$ has
length at least $2i$, we only need to consider indices $i$ such that $2i<\best$.
Moreover (and this is not described in the algorithm, but it is done
in the implementation), we can consider first
the pairs that give walks for length $2i$ first, like for example $L^0_i\times L^1_{i-1}$ 
and then the ones that give length  $2i+1$, like for example $L^0_i\times L^1_i$.
If we use this order, as soon as we find an edge in 
the while-loop, we can finish the work for the root $r$, and move onto the next root.

\section{Implementation and experiments}
\label{sec:implementation}

We have implemented the algorithms of Section~\ref{sec:algorithms}
in C++ using CGAL version~4.6.3~\cite{cgal}
because it provides the more complex procedures we need:
Delaunay triangulations and Voronoi diagrams~\cite{cgal:k-vda2-15a}, 
range trees~\cite{cgal:n-rstd-15a}, and nearest neighboours~\cite{cgal:tf-ssd-15a}.
Although in some cases we had to make small modifications, it was
very helpful to have the CGAL code available as a starting point. 
The coordinates of the points were Cartesian doubles.

Experiments were carried out in a laptop with CPU i7-6700HQ at 2.60 Ghz, 
8GB of RAM, and Windows 10. \emph{All times we report are in seconds.}

\paragraph{Data generation}
Data points were generated uniformly at random in the following polygonal domains:
rectangles without holes, rectangles with a "small" rectangular hole, rectangles
with a "large" rectangular hole, rectangles with 4 "small" rectangular holes, 
and rectangles with 4 "large" rectangular holes. The precise proportions of the domains
with holes are in Figures~\ref{fig:data_generation1} and~\ref{fig:data_generation2}.
We generated 1K, 2K, 5K, 10K, 20K and 50K points for the cases where the outer rectangle
has sizes $4\times 1$, $8\times 2$,\ldots, $128\times 32$.
The data was generated once and stored.
For the minimum-separation problem $s$ was placed in the middle of a hole and 
$t$ vertically above $s$ in the outer face.  
Some of these domains are not meaningful for the minimum-separation problem 
because the disks centered at the points cover $s$. 

\begin{figure}
	\includegraphics[width=\textwidth,page=1]{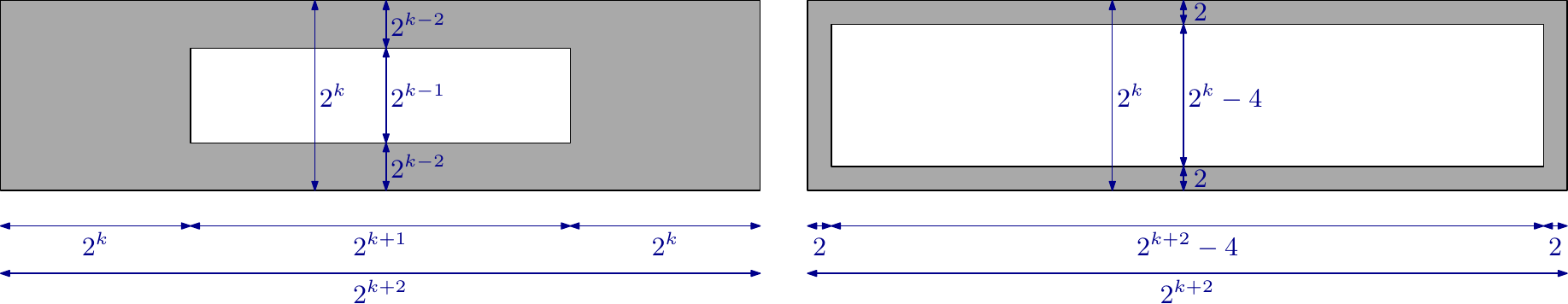}
	\caption{Data generation with a small hole (left) and a large hole (right).}
	\label{fig:data_generation1}
\end{figure}

\begin{figure}
	\includegraphics[width=\textwidth,page=2]{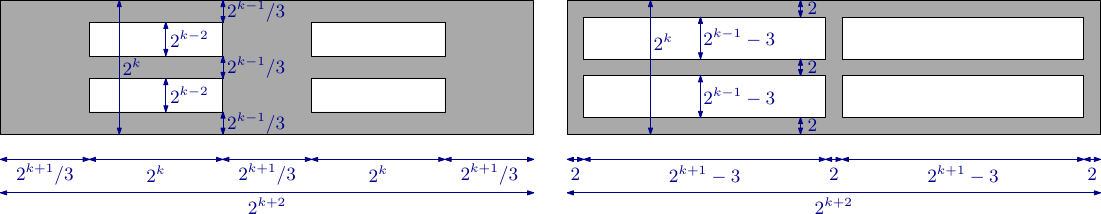}
	\caption{Data generation with four small holes (right) and four large holes (right).}
	\label{fig:data_generation2}
\end{figure}

\paragraph{Shortest-path tree in unit-disk graphs}
We have implemented the algorithm described in Section~\ref{sec:algorithm-sptree}.
For the shortest-path tree we used the Delaunay triangulation as provided by CGAL.
The data structure for nearest neighbour queries is a small extension of
the one provided by~\cite{cgal:k-vda2-15a}, which in turn is based on the Delaunay triangulation.
When making a query for nearest neighbour of $p$ in $W_{i-1}$ (line 17 in Figure~\ref{fig:BFS}), 
we have the option to provide an extra parameter that acts as some sort of hint: 
if the nearest neighbour is near the hint, the algorithm is faster. 
For our implementation, we exploit this as follows. 
Consider an iteration of the while loop (lines 13--22). If the point $q$ is from $W_{i-1}$
then we use a point in a face of $DT(W_{i-1})$ incident to $q$ as the hint
for all the points $p$ considered in the iteration.
If the point $q$ is not from $W_{i-1}$, then we already know that $q\in W_i$ and 
thus $\pi[q]\in W_{i-1}$. In this case we use use a point in a face of $DT(W_{i-1})$
incident to $\pi[q]$
as hint for all the points $p$ considered in the iteration.
Using such hints reduced the running time substantially, 
so we used this feature in the implementation.
Note that this improvement does not come with guarantees in the worst-case.
In the tables we refer to this algorithm as \DEF{SSSP}.

We compared the implementation with two obvious alternative algorithms to
compute shortest-path trees.
The first alternative is to build the graph $G=\GG(P)$ explicitly. 
Thus, for each pair of points $p,q$ we check whether their distance is
at most one and add an edge to a graph data structure.
We can then use breadth-first-search (BFS) from the given root $r$.
The preprocessing is quadratic, and the time spent to compute
a shortest-path tree depends on the density of the graph $G$.
In the tables we refer to this algorithm as \DEF{BFS}.

The second alternative we consider is to use a unit-length grid. 
Two points $(x,y)$ and $(x',y')$ are in the same grid cell
if and only if 
$(\lfloor x\rfloor ,\lfloor y\rfloor)=(\lfloor x'\rfloor ,\lfloor y'\rfloor)$.
We store all the points of a grid cell $c$ in a list $\ell(c)$.
The non-empty lists $\ell(c)$ are stored in a dictionary,
where the bottom-left corner of the cell is used as key.
We can then run some sort of BFS using this structure. 
The list $\ell(c)$ for a cell $c$ maintains the points that
have not been visited by the BFS tree yet. When processing
a point $p$ in a cell $c$, we have to treat all the points 
in the lists of $c$ and its $8$ adjacent cells as candidate points.
Any point that is adjacent to $p$ is then removed from the list of its cell. 
The preprocessing is linear, and the time spent to compute
a shortest-path tree depends on the distribution of the points.
It is easy to produce cases where the algorithm would need quadratic time.
For each shortest-path tree we compute the lists and the dictionary anew.
(This step is very fast in any case.)
In the tables we refer to this algorithm as \DEF{grid}.

As mentioned earlier, we did not implement the algorithm of Chan and Skrepetos~\cite{ChanS16}
because of time constraints. We expect that it would work good.

The measured times are in Tables~\ref{table1}--\ref{table5}.
For SSSP and BFS we report the preprocessing time that is independent of the source
(like building the Delaunay triangulation or building the graph) and the average
time spent for a shortest-path tree over 50 choices of the root.
For grid we just report the total running time; assigning points to
the grid cells and putting them into a dictionary is almost negligible.
As it can be seen, the results for SSSP are very much independent of the shape and,
for dense point sets it outperforms the other algorithms.

While the algorithm SSSP has guarantees in the worst case,
for BFS and grid one can construct instances 
where the behavior will be substantially bad. 
For example, to the instance with 10K points in a rectangle of size $32\times 8$ 
with a small hole we added 1K points quite cluttered. 
The increase in time with respect to the original instance 
was for SSSP 9,7\% (preprocessing) and 13,6\% (one root), for BFS it was
21,9\% (preprocessing) and 56,5\% (one root), and for grid it was 25\%.

\begin{table}[ht]
\begin{tabular}{l*{6}{r}}
\textbf{Rectangle without holes} & \multicolumn{6}{c}{20K points}\\						
size rectangle	&	$4\times 1$	&	$8\times 2$	&	$16\times 4$	&	$32\times 8$	&	$64\times 16$	&	$128\times 32$	\\
\hline
SSSP preprocessing	&	0.018	&	0.018	&	0.018	&	0.018	&	0.019	&	0.021	\\
SSSP average/root	&	0.011	&	0.012	&	0.012	&	0.012	&	0.013	&	0.013	\\
BFS preprocessing	&	18.70	&	13.46	&	12.03	&	11.40	&	11.32	&	11.13	\\
BFS average/root	&	2.437	&	1.018	&	0.321	&	0.069	&	0.017	&	0.005	\\
grid				&	1.309	&	1.130	&	0.474	&	0.160	&	0.060	&	0.035\vspace{.2cm}	\\
\hline
  & \multicolumn{6}{c}{50K points}\\						
\hline
SSSP preprocessing	&	0.051	&	0.050	&	0.053	&	0.051	&	0.051	&	0.053	\\
SSSP average/root	&	0.034	&	0.037	&	0.037	&	0.036	&	0.035	&	0.036	\\
BFS preprocessing	&	$>$2min	&	86.12	&	74.76	&	74.15	&	72.41	&	71.49	\\
BFS average/root	& \hspace{-.7cm}memory limit & 6.524	&	2.422	&	0.510	&	0.119	&	0.035	\\
grid				&	6.297	&	7.125	&	3.188	&	0.923	&	0.301	&	0.139
\end{tabular}
\caption{Times for shortest paths in rectangles without holes.}
\label{table1}
\end{table}

\begin{table}[ht]
\begin{tabular}{l*{6}{r}}
\textbf{Rectangle 1 small hole}& \multicolumn{6}{c}{10K points}\\						
size rectangle	&	$4\times 1$	&	$8\times 2$	&	$16\times 4$	&	$32\times 8$	&	$64\times 16$	&	$128\times 32$	\\
\hline
SSSP preprocessing	&	0.011	&	0.012	&	0.009	&	0.010	&	0.010	&	0.009	\\
SSSP average/root	&	0.004	&	0.005	&	0.005	&	0.005	&	0.006	&	0.006	\\
BFS preprocessing	&	3.724	&	3.033	&	2.890	&	2.826	&	2.874	&	2.841	\\
BFS average/root	&	0.587	&	0.248	&	0.078	&	0.021	&	0.006	&	0.002	\\
grid				&	0.258	&	0.313	&	0.119	&	0.049	&	0.022	&	0.015	\vspace{.2cm}	\\
\hline
  & \multicolumn{6}{c}{20K points}\\						
\hline
SSSP preprocessing	&	0.019	&	0.019	&	0.019	&	0.019	&	0.018	&	0.023	\\
SSSP average/root	&	0.010	&	0.012	&	0.011	&	0.012	&	0.013	&	0.013	\\
BFS preprocessing	&	15.22	&	13.47	&	11.51	&	11.66	&	11.73	&	11.38	\\
BFS average/root	&	2.402	&	1.045	&	0.369	&	0.088	&	0.023	&	0.006	\\
grid				&	1.122	&	1.339	&	0.461	&	0.181	&	0.074	&	0.036
\end{tabular}
\caption{Times for shortest paths in rectangles with a small hole.}
\label{table2}
\end{table}

\begin{table}[ht]
\begin{tabular}{l*{3}{r}|*{3}{r}}
\textbf{Rectangle 1 large hole} & \multicolumn{3}{c|}{5K points} & \multicolumn{3}{c}{10K points}\\
size rectangle	&	$32\times 8$	&	$64\times 16$	&	$128\times 32$	&	$32\times 8$	&	$64\times 16$	&	$128\times 32$\\						
\hline
SSSP preprocessing	&	0.004	&	0.005	&	0.005	&	0.009	&	0.010	&	0.010	\\
SSSP average/root	&	0.002	&	0.002	&	0.002	&	0.005	&	0.005	&	0.005	\\
BFS preprocessing	&	0.751	&	0.767	&	0.742	&	2.783	&	3.175	&	2.804	\\
BFS average/root	&	0.006	&	0.003	&	0.002	&	0.025	&	0.012	&	0.006	\\
grid				&	0.018	&	0.012	&	0.008	&	0.053	&	0.032	&	0.022
\end{tabular}
\caption{Times for shortest paths in rectangles with a large hole.}
\label{table3}
\end{table}

\begin{table}[ht]
\begin{tabular}{l*{3}{r}|*{3}{r}}
\textbf{Rectangle 4 small holes\hspace{-.2cm}} & \multicolumn{3}{c|}{10K points} & \multicolumn{3}{c}{20K points}\\
size rectangle	&	$32\times 8$	&	$64\times 16$	&	$128\times 32$	&	$32\times 8$	&	$64\times 16$	&	$128\times 32$\\						
\hline
SSSP preprocessing	&	0.010	&	0.011	&	0.009	&	0.018	&	0.018	&	0.019	\\
SSSP average/root	&	0.005	&	0.006	&	0.007	&	0.012	&	0.013	&	0.014	\\
BFS preprocessing	&	2.925	&	2.861	&	2.866	&	11.97	&	11.93	&	11.59	\\
BFS average/root	&	0.020	&	0.006	&	0.002	&	0.085	&	0.022	&	0.006	\\
grid				&	0.048	&	0.024	&	0.016	&	0.190	&	0.070	&	0.040
\end{tabular}
\caption{Times for shortest paths in rectangles with 4 small holes.}
\label{table4}
\end{table}

\begin{table}[ht]
\begin{tabular}{l*{3}{r}|*{3}{r}}
\textbf{Rectangle 4 large holes} & \multicolumn{3}{c|}{5K points} & \multicolumn{3}{c}{10K points}\\
size rectangle	&	$32\times 8$	&	$64\times 16$	&	$128\times 32$	&	$32\times 8$	&	$64\times 16$	&	$128\times 32$\\						
\hline
SSSP preprocessing	&	0.004	&	0.005	&	0.005	&	0.013	&	0.015	&	0.009	\\
SSSP average/root	&	0.003	&	0.003	&	0.003	&	0.006	&	0.005	&	0.005	\\
BFS preprocessing	&	0.715	&	0.734	&	0.717	&	2.897	&	2.910	&	3.182	\\
BFS average/root	&	0.005	&	0.002	&	0.001	&	0.019	&	0.008	&	0.004	\\
grid				&	0.013	&	0.010	&	0.008	&	0.045	&	0.026	&	0.020
\end{tabular}
\caption{Times for shortest paths in rectangles with 4 large holes.}
\label{table5}
\end{table}

\paragraph{Minimum separation with unit-disk}
We have implemented the algorithm \proc{GenericMinimumSeparation}
and the new algorithm based on a compact treatment of the edges.
The shortest-path trees are constructed using the algorithm of Section~\ref{sec:algorithm-sptree}.
The table $N[~]$ and the sets $L^0_i,L^1_i,R^0_i,R^1_i$
are constructed at the time of computing the shortest-path tree.

In the data structure of Lemma~\ref{lem:block},
we do use a 2-dimensional tree as the primary structure, making some modifications
of~\cite{cgal:n-rstd-15a}. In the secondary structure, for nearest neighbour, 
instead of using Voronoi diagrams, we used a small modification of
the $kd$-trees implemented in~\cite{cgal:tf-ssd-15a}. 
In some preliminary experiments this seemed to be a better choice.
In our modification, we make a range search query for points at distance at most $1$,
and finish the search whenever we get the first point.
In the new algorithm, before calling to the function to candidates pairs, like for example $L_i^0\times R^i_1$,
we test that both sets are non-empty. This simple test reduced the time by 30-50\%
in our test cases.

Besides the new algorithm we also implemented the generic algorithm of Section~\ref{sec:generic}.
The measured times are in Tables~\ref{table6}--\ref{table7}.
For the case of $4$ holes we always put $t$ above the rectangle and $s$ in one hole. 
It seems that the choice of the hole does not substantially affect the experimental time in our setting.

To show that our new algorithm can work substantially faster than the generic algorithm,
we created an instance where we expect so.
For this we take the rectangle of size $32\times 8$ with one small hole,
the original 2K points, and add 500 extra points on a vertical strip of width $1$ within the domain
and symmetric with respect to segment $st$. The generic algorithm took 435 seconds and the new
algorithm took 94 seconds. If instead we add 1K points, the generic algorithm takes more than 15 minutes and the new algorithm takes 173 seconds.

\begin{table}[ht]
\begin{tabular}{l*{4}{r}}
\textbf{Rectangle 1 small hole} & \multicolumn{4}{c}{2K points}\\
size rectangle	&	$8\times 2$	&	$16\times 4$	&	$32\times 8$ & $64\times 16$ \\	
\hline
new separation algorithm	&	41	&	41	&	32	&	30  \\
generic algorithm			&	730	&	215	&	67	&	30 \\
cycle length				&   9   &   20  &   46  &   126
\end{tabular}
\caption{Times for minimum separation with 1 hole.}
\label{table6}
\end{table}

\begin{table}[ht]
\begin{tabular}{l*{3}{r}|*{3}{r}}
\textbf{Rectangle 4 holes} & \multicolumn{3}{c|}{2K points, small holes} & \multicolumn{3}{c}{5K points, large holes}\\
size rectangle	&	$32\times 8$ &	$64\times 16$ & $128\times 32$ & $32\times 8$ &	$64\times 16$ & $128\times 32$\\	
\hline
new separation algorithm	&	20	&	25	&	5.6	&	233 & 259 & 240  \\
generic algorithm			&	62	&	25	&	6.4	&	875 & 428 & 248 \\
cycle length				&	24  &   61	&	201 &	29	& 77  & 342 
\end{tabular}
\caption{Times for minimum separation with 4 holes.}
\label{table7}
\end{table}


\end{document}